\documentclass[pdflatex, sn-basic, iicol]{sn-jnl}
\usepackage{mathtools, nicefrac, array}
\usepackage{enumitem}
\usepackage{program}
\catcode`\|=12\relax
\usepackage{tikz}
\usepackage{bookmark}
\usepackage{xparse}

\jyear{2022}

\theoremstyle{thmstyleone}
\newtheorem{proposition}{Proposition}
\newtheorem*{corollary}{Corollary}
\theoremstyle{thmstylethree}
\newtheorem*{remark}{Remark}

\newcommand*{\beq}[1]{\begin{equation} \label{#1}}
\newcommand*{\eeq}{\end{equation}}

\renewcommand{\d}{\operatorname{d\!}{}}
\renewcommand*{\Re}{\operatorname{Re}}
\renewcommand*{\Im}{\operatorname{Im}}
\newcommand*{\defeq}{\vcentcolon=}

\DeclarePairedDelimiter{\abs}{\lvert}{\rvert}
\DeclarePairedDelimiterX{\inprod}[2]{\langle}{\rangle}{#1,#2}
\DeclarePairedDelimiterXPP{\bprob}[1]{\mathbb{P}}{[}{]}{}{#1}
\DeclarePairedDelimiterXPP{\parambprob}[1]{\mathbb{P}^\eta}{[}{]}{}{#1}
\DeclarePairedDelimiterXPP{\expect}[1]{\mathbb{E}}{[}{]}{}{#1}
\newcommand*{\given}[1][]{\nonscript\:#1\vert\nonscript\:\mathopen{}}

\DeclareSymbolFont{otherbbold}{U}{bbold}{m}{n}
\DeclareMathSymbol{\mathbbone}{\mathbb}{otherbbold}{"31}
\newcommand*{\charf}[1]{\mathbbone_{#1}}

\NewDocumentCommand{\E}{m o}{\IfNoValueTF{#2}{\mathcal{#1}}{\mathcal{#1}_{#2}}}
\NewDocumentCommand{\dE}{m o}{\IfNoValueTF{#2}{\dot{\mathcal{#1}}}{\dot{\mathcal{#1}}_{#2}}}
\newcommand*{\dom}[2][f]{%
    \ifx#1f
        \mathscr{D}_{#2}%
    \else\ifx#1w
        \mathscr{D}_{#2}^\textsc{wc}%
    \else\fi\fi%
}

\newcommand*{\itembox}[1]{%
    \item \makebox[\widthof{\(0\)}][c]{\(#1\)}%
}

\begin{document}

\title[Beyond Wilson--Cowan dynamics]{Beyond Wilson--Cowan dynamics: oscillations and chaos without inhibition}

\author*[1]{\fnm{Vincent} \sur{Painchaud}}{\email{vincent.painchaud@mail.mcgill.ca}}
\author[2,4,5]{\fnm{Nicolas} \sur{Doyon}}{\email{nicolas.doyon@mat.ulaval.ca}}
\author[3,4,5]{\fnm{Patrick} \sur{Desrosiers}}{\email{patrick.desrosiers@phy.ulaval.ca}}

\affil*[1]{\orgdiv{Department of Mathematics and Statistics}, \orgname{McGill University}, \orgaddress{\street{Sherbrooke Street West}, \city{Montreal}, \postcode{H3A 0B6}, \state{Quebec}, \country{Canada}}}
\affil[2]{\orgdiv{D\'epartment de math\'ematiques et de statistique}, \orgname{Universit\'e Laval}, \orgaddress{\street{avenue de la M\'edecine}, \city{Quebec City}, \postcode{G1V 0A6}, \state{Quebec}, \country{Canada}}}
\affil[3]{\orgdiv{D\'epartment de physique, de g\'enie physique et d'optique}, \orgname{Universit\'e Laval}, \orgaddress{\street{avenue de la M\'edecine}, \city{Quebec City}, \postcode{G1V 0A6}, \state{Quebec}, \country{Canada}}}
\affil[4]{\orgname{CERVO Brain Research Center}, \orgaddress{\street{avenue d'Estimauville}, \city{Quebec City}, \postcode{G1E 1T2}, \state{Quebec}, \country{Canada}}}
\affil[5]{\orgname{Centre interdisciplinaire en mod\'elisation math\'ematique de l'Universit\'e Laval},
\orgaddress{\street{avenue de la M\'edecine},
\city{Quebec City}, \postcode{G1V 0A6}, \state{Quebec}, \country{Canada}}}

\abstract{Fifty years ago, Wilson and Cowan developed a mathematical model to describe the activity of neural populations. In this seminal work, they divided the cells in three groups: active, sensitive and refractory, and obtained a dynamical system to describe the evolution of the average firing rates of the populations. In the present work, we investigate the impact of the often neglected refractory state and show that taking it into account can introduce new dynamics. Starting from a continuous-time Markov chain, we perform a rigorous derivation of a mean-field model that includes the refractory fractions of populations as dynamical variables. Then, we perform bifurcation analysis to explain the occurance of periodic solutions in cases where the classical Wilson--Cowan does not predict oscillations. We also show that our mean-field model is able to predict chaotic behavior in the dynamics of networks with as little as two populations.}

\keywords{Biological neural networks, Wilson--Cowan model, dynamical systems, Markov chains, chaos}

\maketitle

\section{Introduction}

Differential equations have been successfully used to model the activity of neurons for more than a century now, ever since the works of \cite{lapicque_recherches_1907}. One of the most important examples of such a model, published fifty years ago by \cite{wilson_excitatory_1972}, describes the average firing rates of coupled neural populations using a set of ordinary differential equations. This model, of significant historical importance, has been the starting point for many extensions and is still highly relevant today \citep{bressloff_stochastic_2016, chow_before_2020, cowan_wilsoncowan_2016, destexhe_wilson-cowan_2009, wilson_evolution_2021}. An important achievement of this model was to predict oscillations and bistability in the neural activity of a network made of an excitatory and an inhibitory population.

Inhibition is a key ingredient in the dynamics of the Wilson--Cowan model. Indeed, it is well known that their equations for a pair of populations can only predict oscillatory solutions if one of them is excitatory while the other is inhibitory \cite[section~11.3.2]{ermentrout_mathematical_2010}, forming a so-called Wilson--Cowan oscillator. The model may also lead to more complicated dynamical behavior in the presence of inhibition. For instance, chaotic behaviors have been shown to arise from their equations, for example by \cite{borisyuk_dynamics_1995} and by \cite{maruyama_analysis_2014}, but only---at least to our knowledge---in systems of at least two coupled Wilson--Cowan oscillators. Thus, inhibition plays a crucial role in Wilson--Cowan's model in the formation of oscillatory and chaotic behaviors, which have both been observed experimentally for a long time in the activity of neural networks and are thought to play important biological roles \citep{buzsaki_rhythms_2006, breakspear_dynamic_2017, rabinovich_role_1998}. In similar models as well, the role of inhibition is crucial for chaotic behavior to arise  \citep{fukai_asymmetric_1990, sompolinsky_chaos_1988}. While it has been demonstrated experimentally that inhibition is important to several types of oscillations in neural networks \citep{bartos_synaptic_2007, whittington_inhibition-based_2000}, it does not explain all oscillatory behavior. For example, it has been shown in experiments that excitatory neurons in the pre-Bötzinger complex can exhibit oscillatory behavior \citep{butera_models_1999-1, butera_models_1999, duan_dynamics_2017}. Another example of this are theta oscillations \citep{buzsaki_theta_2002, buzsaki_neuronal_2004}, which are thought to arise as a result of the activity of excitatory neurons only \citep{budd_theta_2005, chagnac-amitai_synchronized_1989}. Thus, there are still oscillatory behaviors that the classical model cannot explain.

One of the key steps in the construction of Wilson--Cowan's model is a time coarse graining, which has the effect of setting the refractory fraction of a neural population as always proportional to its active fraction. Wilson and Cowan originally argued that, at least when the parameters of the model have physiologically reasonable values, this should not have any important impact on the model. Later on, others even argued that the refractory period of neurons should not have an impact on the dynamics, and chose to neglect it completely \citep[section~11.3]{curtu_oscillations_2001, ermentrout_mathematical_2010}. However, it has been noticed experimentally that refractoriness can have an impact on neuronal activity. For instance, \cite{berry_refractoriness_1998} have shown that a longer refractory period in neurons improves the precision in the response of ganglion cells, suggesting that refractoriness improves neural signaling. Similarly, \cite{avissar_refractoriness_2013} have shown that refractoriness enhances precision in the timing and synchronization of neurons' spikes, which once again suggests that it makes neurons more precise. Refractoriness is also known to be related to oscillations in neural networks, since it can help neurons to synchronize their spikes \citep{sanchez-vives_cellular_2000, wiedemann_timing_2003}. In theoretical works as well, it has recently been suggested that the refractory state of neurons could be essential in some cases to provide a complete description of the dynamics of a biological neural network \citep{rule_neural_2019, weistuch_refractory_2021}.

In this paper, we propose a simple extension of Wilson--Cowan's classical model where the refractory state is explicitly considered with the same importance as the active and sensitive states. Indeed, using a different approach than Wilson--Cowan, we derive a dynamical system closely related to their original model, in the sense that it includes it as a subsystem. Then, we show that our dynamical system can predict oscillations and chaotic behaviors in the activity of neural networks of excitatory populations. This contrasts with the original Wilson--Cowan model, in which inhibition plays a crucial role for such phenomena to arise.

First, in section~\ref{sec.model}, we present an explicit construction of the model. We start by defining a continuous-time Markov chain to describe the evolution of a large network's state in a way that mimics the behavior of biological neurons. The resulting stochastic process is similar to one already proposed---but not extensively studied---by \cite{cowan_stochastic_1990}, and is also reminiscent of a similar process proposed by \cite{zarepour_universal_2019}. Then, we reduce this Markov chain, which describes the evolution of the network's state from a microscopic point of view, to a dynamical system of small dimension that describes the dynamics from a macroscopic point of view. To do so, we split the network into a small number of neural populations, and we obtain a dynamical system that describes the evolution of the average active and refractory fractions of each population. 

Then, in section~\ref{sec.wilsoncowan}, we study the relationship between our model and Wilson--Cowan's classical model. In fact, we show that Wilson--Cowan's dynamical system can be seen as a subsystem of ours. We then argue that the simplification of our model to Wilson--Cowan's is not trivial. Indeed, the domain to which corresponds the subsystem is not invariant by the flow of the full dynamical system, so that in particular, it cannot be an attracting set. 

Finally, in section~\ref{sec.examples}, we present a detailed study of three examples where our model succeeds in predicting the qualitative dynamical behavior of the underlying Markov chain, while the classical Wilson--Cowan model fails to do so. In particular, the first example shows that our model allows oscillations in the activity of a single excitatory population, and the second shows that it allows chaotic behavior in the activity of a pair of excitatory populations.

\section{The model}
\label{sec.model}

We consider a network of \(N\) neurons labelled with integers from \(1\) to \(N\). Links between neurons are described by a weight matrix \(W \in \mathbb{R}^{N\times N}\) whose element \(W_{jk}\) represents the weight of the connection from neuron \(k\) to neuron \(j\). If neuron \(k\) is excitatory, then \(W_{jk} > 0\), and if it is inhibitory, then \(W_{jk} < 0\). 

Our goal is to build an approximate macroscopic description of the network's dynamics starting from a precise microscopic description. To do so, we consider a partition \(\mathscr{P}\) of the set of neurons \(\{1, \hdots, N\}\). Each element \(J \in \mathscr{P}\) then represents a population of the network, that is, a set of neurons that share similar properties in a sense that will be made precise later. We will start by defining a continuous-time Markov process to provide a precise description of the evolution of the network's state, and then we will construct a mean-field model to describe the evolution of the population's macroscopic states.

\subsection{The microscopic model}
\label{sec.micromodel}

\subsubsection{Stochastic process}
\label{sec.markovchain}

In order to model biological neurons, we assume that neurons can take three states:
\begin{itemize}
    \itembox{0} : the \emph{sensitive} state,
    \itembox{1} : the \emph{active} state,
    \itembox{i} : the \emph{refractory} state,
\end{itemize}
where \(i\) denotes the imaginary unit. The active state is that of a neuron undergoing an action potential. Following an action potential, neurons typically enter a hyperpolarized state during which another action potential cannot occur even if the neuron receives a stimulus that would be otherwise sufficient to trigger spiking \citep{purves_neuroscience_2018}. This state is what we call the refractory state. When a neuron is neither active nor refractory, we say that it is sensitive, as if it receives a high enough input it can spike in response.

We see the transitions between these states as random: a sensitive neuron activates at a rate that increases nonlinearly with its input, and then gets to the refractory state and then back to the sensitive state at constant rates. This intuitive process describes the evolution of the whole network's state, and is defined rigorously as a continuous-time Markov chain \(\{X_t\}_{t\geq 0}\) taking values in the state space \(E \defeq \{0,1,i\}^N\). 

\begin{figure}
\centering
\begin{tikzpicture}
    \node (0) at (-30:1) {\(0\)};
    \node (1) at ( 90:1) {\(1\)};
    \node (i) at (210:1) {\(i\)};
    \draw[-stealth] (0) to[bend right=42] node [midway,above] {\(\alpha\)} (1);
    \draw[-stealth] (1) to[bend right=42] node [midway,above] {\(\beta\)} (i);
    \draw[-stealth] (i) to[bend right=42] node [midway,above] {\(\gamma\)} (0);
\end{tikzpicture}
\caption{Allowed transitions between states of neurons and corresponding characteristic rates.}
\label{fig.rates}
\end{figure}
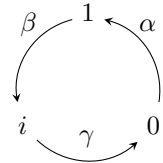

For each neuron \(j\), let \(\alpha_j, \beta_j, \gamma_j > 0\). These parameters characterize the transition rates from a state to another as illustrated in Fig.~\ref{fig.rates}. While \(\beta_j\) and \(\gamma_j\) both describe the actual transition rates, \(\alpha_j\) rather represents the activation rate of \(j\) if it is given an infinite excitation. Indeed, we assume a soft threshold dynamics: the activation rate is given by the function
\[
a_j(x) \defeq \alpha_j F_J\biggl( \sum_{k=1}^N W_{jk} \Re x_k + Q_J \biggr),
\]
where \(F_J\colon \mathbb{R} \to [0,1]\) is a function of the neuron's input, \(Q_J \in \mathbb{R}\) is an input in \(J\) that is external to the network, and \(J\) is the population to which belongs \(j\). Here, we assume that \(F_J\) is a continuous and increasing function that tends to \(1\) at infinity. Biologically speaking, \(\beta_j\) and \(\gamma_j\) can be interpreted as the inverse of the average times neuron \(j\) spends in the active and refractory states, respectively. Thus, for example, a higher value for \(\beta_j\) than for \(\gamma_j\) translates the idea that the duration of an action potential is less than the refractory period. However, the interpretation of \(\alpha_j\) is less direct. Indeed, \(\alpha_j\) is the activation rate of neuron \(j\) only when it is given an infinite input. Hence, the inverse of \(\alpha_j\) cannot be seen directly as the average time neuron \(j\) spends in the sensitive state; this time depends on the activity of the whole network. Moreover, the average and maximal firing rates vary greatly from situation to situation or according to the neural type~\citep{roxin_distribution_2011, wang_firing_2016}. This implies that the model can be relevant for a wide range of values of \(\alpha_j\). 

From the rates associated to the state transitions of single neurons, we can now define a generator for the Markov chain \(\{X_t\}_{t\geq 0}\), which will allow to define the process correctly. This generator is a matrix \(M\) indexed over the state space \(E\) whose entry \(m(x,y)\) gives the transition rate from a state \(x\) to another state \(y\). We define this rate as
\[
m(x,y) \defeq \sum_{j=1}^N m_j(x,y) \smash{\prod_{\substack{k=1 \\ k \neq j}}^N} \delta_{x_ky_k},
\]
where \(\delta_{x_ky_k}\) is a Kronecker delta and where
\[
\begin{aligned}
m_j(x,y) \defeq a_j(x) & (1 - \abs{x_j}) \bigl( \Re y_j - (1 - \abs{y_j}) \bigr) \\
    {} + \beta_j & \Re(x_j) (\Im y_j - \Re y_j) \\
    {} + \gamma_j & \Im(x_j) \bigl( (1 - \abs{y_j}) - \Im y_j \bigr).
\end{aligned}
\]
To make sense of \(m_j(x,y)\), recall that a component \(x_j\) of a state vector \(x \in E\) is either \(0\), \(1\) or \(i\), so that exactly one of \(\Re x_j\), \(\Im x_j\) and \(1 - \abs{x_j}\) is \(1\) while the others are \(0\).

Now, it is a simple calculation to see that the matrix \(M \defeq \{m(x,y) : x,y \in E\}\) is the generator of a continuous-time Markov chain. Thus, it follows from the Kolmogorov extension theorem (for details, see e.g. the books by \cite{doob_stochastic_1990} or \cite{norris_markov_1997}) that a probability measure \(\mathbb{P}\) exists on the space \((E,2^E)^{[0,\infty)}\) such that for any \(x,y \in E\), as \(\Delta t \downarrow 0\),
\[
\bprob{X_{t+\Delta t} = y \given X_t = x} = \delta_{xy} + m(x,y) \Delta t + o(\Delta t).
\]
In particular, if \(X_t^j\) denotes the \(j\)th component of \(X_t\) and if \(x\) denotes a state with \(x_j = 0\), then
\begin{subequations}
\label{eq.transitionprob}
\begin{align}
\bprob{X_{t+\Delta t}^j = 1 \given X_t = x} & = a_j(x)\Delta t + o(\Delta t),
\intertext{while in general}
\bprob{X_{t+\Delta t}^j = i \given X_t^j = 1} & = \beta_j\Delta t + o(\Delta t), \\
\bprob{X_{t+\Delta t}^j = 0 \given X_t^j = i} & = \gamma_j\Delta t + o(\Delta t),
\end{align}
\end{subequations}
and the transitions \(0 \mapsto i\), \(i \mapsto 1\) and \(1 \mapsto 0\) all have \(o(\Delta t)\) rates.

In principle, the system is then completely described, since the transition probabilities \(\bprob{X_{t+\Delta t} = y \given X_t = x}\) for \(x,y \in E\) can be obtained from the solution of the Kolmogorov forward equation
\beq{eq.kolmogorov}
\dot{P}(t) = P(t)M,
\eeq
where the dot denotes a derivative. Indeed, the solution \(P(t)\) of~\eqref{eq.kolmogorov} is the matrix whose entries are the probabilities that the system makes a transition from a state to another during an interval of time \(t\). More explicitely, the transition probability \(\bprob{X_{t+\Delta t} = y \given X_t = x}\) is equal to the element \(x,y\) of \(P(\Delta t)\). However, since there are \(3^N\) possible states in \(E\), this differential equation is enormous when the network has a large number of neurons, so that it cannot be studied directly in practice.

\subsubsection{Dynamical system}
\label{sec.dynsys}

Now, we want to use the stochastic process constructed above to obtain a macroscopic description of the evolution of the network's state, in the form of a dynamical system. To do this, we first introduce functions \(p_j, q_j, r_j \colon [0,\infty) \to [0,1]\) given by
\begin{subequations}
\begin{align}
\SwapAboveDisplaySkip
    p_j(t) & \defeq \bprob{X_t^j = 1}, \\
    q_j(t) & \defeq \bprob{X_t^j = 0}, \\
    r_j(t) & \defeq \bprob{X_t^j = i}.
\end{align}
\end{subequations}
Since \(X_t^j\) takes values in \(\{0,1,i\}\), it is easy to see that \(p_j + q_j + r_j \equiv 1\) and that
\begin{align*}
    p_j(t) & = \expect{\Re X_t^j}, \\
    q_j(t) & = \expect{1 - \abs{X_t^j}}, \\
    r_j(t) & = \expect{\Im X_t^j},
\end{align*}
where \(\mathbb{E}\) denotes the expectation with respect to \(\mathbb{P}\). Using these relations, it is possible to find expressions for the derivatives of these variables. Indeed, with \(\Delta t > 0\),
\begin{align*}
p_j(t + \Delta t) 
    & = \bprob{X_{t+\Delta t}^j = 1} \\
    & = \sum_{x\in E} \bprob{X_{t+\Delta t}^j = 1 \given X_t = x} \bprob{X_t = x}.
\intertext{Using the transition rates introduced earlier, we get from~\eqref{eq.transitionprob} that}
\begin{split}
p_j(t + \Delta t)
    & = \sum_{x\in E} \Bigl( 
        \Re(x_j)(1 - \beta_j \Delta t) \\[-2mm] 
        & \hspace*{12mm} + \bigl( 1 - \abs{x_j} \bigr) a_j(x) \Delta t \\
        & \hspace*{22mm} + o(\Delta t) 
    \Bigr) \bprob{X_t = x}
\end{split} \\
\begin{split}
    & = (1 - \beta_j \Delta t) p_j(t) \\ 
    & \hspace*{6mm} + \Delta t \expect[\big]{\bigl( 1 - \abs{X_t^j} \bigr) a_j(X_t)} + o(\Delta t).
\end{split}
\end{align*}
Taking \(\Delta t \to 0\), it follows that
\begin{subequations}
\label{eq.microsystem}
\begin{align}
\dot{p}_j(t) & = - \beta_j p_j(t) + \expect[\big]{\bigl( 1 - \abs{X_t^j} \bigr) a_j(X_t)}.
\intertext{Using the same method, we also find that}
\dot{q}_j(t) & = - \expect[\big]{\bigl( 1 - \abs{X_t^j} \bigr) a_j(X_t)} + \gamma_j r_j(t), \\
\dot{r}_j(t) & = - \gamma_j r_j(t) + \beta_j p_j(t).
\end{align}
\end{subequations}
At first glance, one might think that at this point, the dimension of the system has been reduced from \(3^N\) to \(2N\). However, the system~\eqref{eq.microsystem} is not closed, in the sense that the derivatives of \(p_j\), \(q_j\) and \(r_j\) are not given as functions of the same variables, but rather involve other expectations. Hence, the Kolmogorov forward equation~\eqref{eq.kolmogorov} is still needed to solve~\eqref{eq.microsystem}.

\subsection{The macroscopic model}
\label{sec.macromodel}

To study the network's dynamics from a macroscopic point of view, we introduce for each population \(J \in \mathscr{P}\) and each \(t \geq 0\) the random variables%
\begin{subequations}
\begin{align}
A_t^J & \defeq \frac{1}{\abs{J}} \sum_{j\in J} \Re X_t^j, \\
R_t^J & \defeq \frac{1}{\abs{J}} \sum_{j\in J} \Im X_t^j, \\
S_t^J & \defeq \frac{1}{\abs{J}} \sum_{j\in J} \bigl( 1 - \abs{X_t^j} \bigr),
\end{align}
\end{subequations}
which can be understood as state variables for populations. Thus, the expected values of these variables describe the expected behavior of the network from a macroscopic point of view. We will use the system~\eqref{eq.microsystem} to find a dynamical system that describes the expected behavior of the network in that sense.

Since it drastically simplifies the reduction of~\eqref{eq.microsystem} to the macroscopic point of view, we assume that all parameters (the weights and transition rates) are constant over populations. However, we stress that this is not necessary if populations are large: we could instead assume that parameters are independent random variables, identically distributed over populations, and the resulting macroscopic model would be the same. This more general approach, which is also much more technical, is discussed in appendix~\ref{sec.apx.construction}.

Using the assumption described above, we see that the input in a neuron \(j\) of a population \(J\) becomes
\[
\begin{aligned}
\sum_{k=1}^N W_{jk} & \Re X_t^k + Q_J \\
    & = \sum_{K\in\mathscr{P}} \sum_{k\in K} W_{JK} \Re X_t^k + Q_J \\
    & = \sum_{K\in\mathscr{P}} \abs{K} W_{JK} A_t^K + Q_J,
\end{aligned}
\]
where we replaced \(W_{jk}\) with \(W_{JK}\) for \(k \in K\), as we assume that weights are constant over populations. To simplify notation, we will write such an input as
\[
B_t^J \defeq \sum_{K\in\mathscr{P}} c_{JK} A_t^K + Q_J
\quad\text{with}\quad
c_{JK} \defeq \abs{K} W_{JK}.
\]
It is then easy to obtain expressions for the derivatives of the average macroscopic state variables%
\begin{subequations}
\begin{align}
    \E{A}[J](t) & \defeq \expect{A_t^J}, \\
    \E{R}[J](t) & \defeq \expect{R_t^J}, \\
    \E{S}[J](t) & \defeq \expect{S_t^J}.
\end{align}
\end{subequations}
Indeed, if \(j \in J\), then \(a_j(X_t) = \alpha_J F_J(B_t^J)\), so it directly follows from~\eqref{eq.microsystem} that
\begin{subequations}
\label{eq.macrosystem}
\begin{align}
    \dE{A}[J](t) & = - \beta_J \E{A}[J](t) + \alpha_J \expect{F_J(B_t^J) S_t^J}, \\
    \dE{S}[J](t) & = - \alpha_J \expect{F_J(B_t^J) S_t^J} + \gamma_J \E{R}[J](t), \\
    \dE{R}[J](t) & = - \gamma_J \E{R}[J](t) + \beta_J \E{A}[J](t),
\end{align}
\end{subequations}
where we replaced \(\alpha_j\) with \(\alpha_J\) and followed the same pattern for other transition rates.

Finally, to close the above dynamical system, we add the mean-field assumption and we neglect covariances between state variables. Thus, we obtain the mean-field dynamical system
\begin{subequations}
\label{eq.meanfield}
\begin{align}
    \dE{A}[J] & = - \beta_J \E{A}[J] + \alpha_J F_J(\E{B}[J]) \E{S}[J], \\
    \dE{S}[J] & = - \alpha_J F_J(\E{B}[J]) \E{S}[J] + \gamma_J \E{R}[J], \\
    \dE{R}[J] & = - \gamma_J \E{R}[J] + \beta_J \E{A}[J]
\end{align}
\end{subequations}
where \(\E{B}[J](t) \defeq \expect{B_t^J}\).

Remark that for each \(J\), one of the equations~\eqref{eq.meanfield} is redundant since \(\E{A}[J] + \E{R}[J] + \E{S}[J] \equiv 1\). Hence, the above dynamical system has dimension \(2n\), where \(n\) is the number of populations of the network. Moreover, the variables \(\E{A}[J]\), \(\E{R}[J]\) and \(\E{S}[J]\) must be contained in \([0,1]\) to make sense. Therefore, the system~\eqref{eq.meanfield} can be studied using only the active and refractory fractions of populations, on the domain
\beq{}
\begin{multlined}
\dom{n} \defeq \bigl\{ (\E{A}[J],\E{R}[J])_{J\in\mathscr{P}} \in [0,1]^{2n} : \\
    \qquad\qquad \forall J \in \mathscr{P}, 0 \leq \E{A}[J] + \E{R}[J] \leq 1 \bigr\}.
\end{multlined}
\eeq

This domain enjoys a simple invariance property.

\begin{proposition}
\label{prop.invariance}
The domain \(\dom{n}\) is invariant by the flow of the dynamical system~\eqref{eq.meanfield}.
\end{proposition}

\begin{proof}
Recall that all transition rates \(\alpha_J, \beta_J, \gamma_J\) are positive and that all functions \(F_J\) are nonnegative. If \(Y = (\E{A}[J],\E{R}[J])_{J\in\mathscr{P}}\) is a point at the boundary of \(\dom{n}\), then one of the fractions \(\E{A}[J]\), \(\E{R}[J]\) or \(\E{S}[J]\) is zero at \(Y\) for some population \(J\); call this fraction \(\E{X}[J]\). Then it is clear from the equations~\eqref{eq.meanfield} that the derivative of \(\E{X}[J]\) must be nonnegative at \(Y\), so that the vector field corresponding to the dynamical system is directed inwards \(\dom{n}\) at \(Y\).
\end{proof}

This invariance property is crucial to the meaning of solutions of the mean-field dynamical system~\eqref{eq.meanfield}. Indeed, for the variables \(\E{A}[J]\), \(\E{R}[J]\) and \(\E{S}[J]\) to represent proportions of neurons, all of them must remain in the interval \([0,1]\) at all times. Thus, Proposition~\ref{prop.invariance} confirms that it is always possible to interpret the components of a solution of~\eqref{eq.meanfield} as proportions of neurons in each of the three states, as long as the initial state can also be interpreted in this way. In other words, any solution of~\eqref{eq.meanfield} that starts from a physiologically meaningful initial state continues to carry a physiological interpretation at all times.

\section{Relationship with the Wilson--Cowan model}
\label{sec.wilsoncowan}

The dynamical system~\eqref{eq.meanfield} can be seen as a generalization of the classical system introduced by \cite{wilson_excitatory_1972}. Indeed, both our model and Wilson--Cowan's describe the activities of populations of a biological neural network. Wilson--Cowan's classical equations, which are formulated for a network split into an excitatory and an inhibitory population, are given by
\begin{subequations}
\label{eq.originalwilsoncowan}
\begin{align}
    \tau_e \dot{E} & = - E + (1 {-} r_e E) f_e(w_{ee} E {-} w_{ei} I {+} Q_e),\!\! \\
    \tau_i \dot{I} & = - I + (1 {-} r_i I) f_i(w_{ie} E {-} w_{ii} I {+} Q_i),
\end{align}
\end{subequations}
where \(E\) and \(I\) are the average firing rates of the populations, \(\tau_e\) and \(\tau_i\) are time constants, \(r_e\) and \(r_i\) are the refractory periods of neurons, \(f_e\) and \(f_i\) are functions that describe the response of both populations, \(Q_e\) and \(Q_i\) are external inputs, and \(w_{JK}\) are nonnegative coefficients that describe the links between the populations.

To obtain these equations, Wilson and Cowan use a time coarse graining that leads to see the refractory fraction of a population as proportional to its active fraction. Remark that in our model, such a reduction amounts to fixing each refractory fraction \(\E{R}[J]\) to its equilibrium solution in the system~\eqref{eq.meanfield}:
\[
\dE{R}[J] = 0 \iff \E{R}[J] = \frac{\beta_J}{\gamma_J} \E{A}[J].
\]
Using this simplification, the system~\eqref{eq.meanfield} simply becomes
\beq{eq.wilsoncowan}
\begin{aligned}
\dE{A}[J] & = - \beta_J \E{A}[J] \\ 
    & \qquad + \alpha_J \biggl( 1 - \Bigl( 1 + \frac{\beta_J}{\gamma_J} \Bigr) \E{A}[J] \biggr) F_J(\E{B}[J]),
\end{aligned}
\eeq
and is completely equivalent to Wilson--Cowan's in the case of a pair of populations. Indeed, we can introduce \(x_J \defeq \beta_J \E{A}[J]\), which is the average proportion of active neurons per unit time---that is, the firing rate---of population \(J\). Then, the above equation leads to
\[
\frac{1}{\beta_J} \dot{x}_J = - x_J + \biggl( 1 - \Bigl( \frac{1}{\beta_J} + \frac{1}{\gamma_J} \Bigr) x_J \biggr) \alpha_J F_J(\E{B}[J]),
\]
where we can write \(\E{B}[J] = \sum_{K\in\mathscr{P}} \frac{c_{JK}}{\beta_K} x_K + Q_J\). The parallel with Wilson--Cowan's original model~\eqref{eq.originalwilsoncowan} is then quite clear if we set \(\tau_J = \nicefrac{1}{\beta_J}\), \(r_J = \nicefrac{1}{\beta_J} + \nicefrac{1}{\gamma_J}\), \(f_J = \alpha_J F_J\) and \(w_{JK} = \nicefrac{c_{JK}}{\beta_K}\).

The above arguments show that Wilson--Cowan's dynamical system can be seen as the subsystem~\eqref{eq.wilsoncowan} of the mean-field system~\eqref{eq.meanfield}, which corresponds to the domain
\[
\dom[w]{n} \defeq \Bigl\{ (\E{A}[J], \E{R}[J])_{J\in\mathscr{P}} \in \dom{n} : \forall J, \E{R}[J] = \frac{\beta_J}{\gamma_J} \E{A}[J] \Bigr\}.
\]
It is then justified to ask whether the reduction of the system~\eqref{eq.meanfield} to the subsystem~\eqref{eq.wilsoncowan} leads to a loss of richness in the dynamics. 

A first hint that dynamical behaviors can be lost by the reduction of~\eqref{eq.meanfield} to the subsystem~\eqref{eq.wilsoncowan} is given in the following proposition.

\begin{proposition}
\label{prop.noninvariance}
The domain \(\dom[w]{n}\) is not invariant by the flow of the dynamical system~\eqref{eq.meanfield}.
\end{proposition}

\begin{proof}
Let \(f\colon\dom{n} \to \mathbb{R}^{2n}\) denote the vector field corresponding to the differential equation~\eqref{eq.meanfield}. Notice that \(\dom[w]{n}\) is a subset of an euclidean subspace of dimension \(n\) in \(\mathbb{R}^{2n}\). Thus, for \(\dom[w]{n}\) to be invariant by the flow of~\eqref{eq.meanfield}, \(f\) must be tangent to \(\dom[w]{n}\) at each of its points. However, this cannot be the case, because on \(\dom[w]{n}\), \(f\) has no component along the axes associated with the refractory fractions, whereas \(\dom[w]{n}\) is not orthogonal to these axes.

To be more explicit, fix a population \(K\in\mathscr{P}\), and choose \(x \in (0,1)\) such that
\[
\frac{\alpha_K\gamma_K}{\alpha_K\beta_K + \alpha_K\gamma_K + \beta_K\gamma_K} < x < \frac{\gamma_K}{\beta_K + \gamma_K},
\]
which is always possible since \(\alpha_K, \beta_K, \gamma_K > 0\). Now, define the vectors \(Y = (\E{A}[J], \E{R}[J])_{J\in\mathscr{P}}\) and \(Y^\perp = (\E{A}[J]^\perp, \E{R}[J]^\perp)_{J\in\mathscr{P}}\) by setting
\[
\begin{aligned}
    \E{A}[K] & = x, &
    \E{R}[K] & = \frac{\beta_K}{\gamma_K} x, \\
    \E{A}[K]^\perp & = 1, &
    \E{R}[K]^\perp & = - \frac{\gamma_K}{\beta_K},
\end{aligned}
\]
and then \(\E{A}[J] = \E{R}[J] = \E{A}[J]^\perp = \E{R}[J]^\perp = 0\) for all \(J \neq K\). Then \(Y \in \dom[w]{n}\) and \(Y^\perp\) is orthogonal to \(\dom[w]{n}\). An example of a phase plane with \(Y\) and \(Y^\perp\) is illustrated in Fig.~\ref{fig.phaseplane}, in a simple case where the network has a single population.

\begin{figure}
\includegraphics[width=\linewidth]{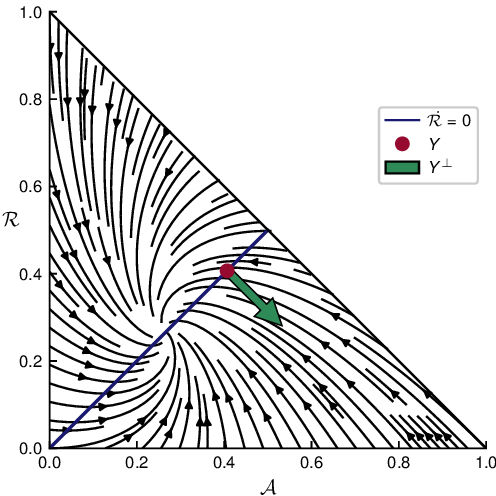}
\caption{Example of phase plane of the mean-field system~\eqref{eq.meanfield} in a case of a single population, to illustrate the proof of Proposition~\ref{prop.noninvariance}. Note that the nullcline \(\dE{R} = 0\) is the set \(\dom[w]{1}\)}
\label{fig.phaseplane}
\end{figure}

Suppose \(f\) is tangent to \(\dom[w]{n}\) at \(Y\). Then \(f(Y)\) must be orthogonal to \(Y^\perp\). Now, all components of \(Y^\perp\) are zero except along the axes of population \(K\), and \(\dE{R}[K]\bigr\rvert_Y = 0\) because \(Y \in \dom[w]{n}\). Thus, \(\inprod{f(Y)}{Y^\perp} = \dE{A}[K]\bigr\rvert_Y\). Since \(F_K\) is nonnegative, we can easily compute the estimate
\begin{align*}
\dE{A}[K]\Bigr\rvert_Y
    & = - \beta_K x + \alpha_K F_K(\E{B}[K]) \E{S}[K] \\
    & \leq - \beta_K x + \alpha_K \Bigl( 1 - \frac{\beta_K + \gamma_K}{\gamma_K} x \Bigr) \\
    & = \alpha_K \Bigl( 1 - \frac{\alpha_K\beta_K + \alpha_K\gamma_K + \beta_K\gamma_K}{\alpha_K\gamma_K} x \Bigr).
\end{align*}
By our choice of \(x\), it follows that \(\inprod{f(Y)}{Y^\perp} < 0\), which contradicts the requirement for \(f(Y)\) to be orthogonal to \(Y^\perp\). Hence, \(f\) is not tangent to \(\dom[w]{n}\) at \(Y\), so \(\dom[w]{n}\) is not invariant by the flow of the dynamical system~\eqref{eq.meanfield}.
\end{proof}

Since an attracting set is invariant by definition, Proposition~\ref{prop.noninvariance} directly implies the following corollary.

\begin{corollary}
The domain \(\dom[w]{n}\) is not an attracting set of the dynamical system~\eqref{eq.meanfield}, so its subsystem~\eqref{eq.wilsoncowan} is not stable.
\end{corollary}

Proposition~\ref{prop.noninvariance} and its corollary give information about the relationship between the solutions of the mean-field system~\eqref{eq.meanfield} and its Wilson--Cowan subsystem~\eqref{eq.wilsoncowan}. If any solution of the mean-field system did eventually converge to a solution of the Wilson--Cowan subsystem---which would happen if \(\dom[w]{n}\) was an attracting set from \(\dom{n}\)---then we would expect that both models actually lead to the same predictions. However, the corollary to Proposition~\ref{prop.noninvariance} shows that this is not true, so that we can expect to see different long-term behaviors predicted by the two models, at least in some cases. In fact, Proposition~\ref{prop.noninvariance} shows that in general, a solution of the Wilson--Cowan subsystem is not even a solution of the mean-field system. This means that if a solution of the mean-field system meets the condition that \(\E{R}[J](t) = \frac{\beta_J}{\gamma_J} \E{A}[J](t)\) for all populations \(J\) at some time \(t\), this condition does not need to remain true afterwards, and the behavior predicted from that point by the Wilson--Cowan subsystem might not be the same as the behavior predicted by the mean-field system.

Nevertheless, the systems~\eqref{eq.meanfield} and~\eqref{eq.wilsoncowan} share some properties: for instance, it is easy to see that their fixed points are exactly the same, since the only way for \(\dE{R}[J] = 0\) is that \(\E{R}[J] = \frac{\beta_J}{\gamma_J} \E{A}[J]\). However, it is not clear at all if the stability of these fixed point remains the same. In fact, we will see in the next section that it is not always the case.

An important difference between these two models is the disparity in the dimension of the dynamical system for the same number of populations. Indeed, in order to model the dynamics of a network of \(n\) neural populations, our mean-field system uses a system of \(2n\) equation whereas the Wilson--Cowan system only uses \(n\). In particular, this means that our mean-field system has enough dimensions to allow oscillations in the activity of a single neural population, unlike the Wilson--Cowan system. Indeed, we will see in section~\ref{sec.ex.cycle1} an example of an excitatory population whose activity oscillates due to the refractory period. However, it is still not possible with our model to predict oscillations in the activity of a single inhibitory population.

\begin{proposition}
\label{prop.nocyclesforinhibitors}
In the case of a single population, if \(\beta + \gamma > \alpha c \sup F'\), then the mean-field system~\eqref{eq.meanfield} has no cycles in the domain \(\dom{1}\).
\end{proposition}

\begin{remark}
Since the transition rates \(\alpha\), \(\beta\) and \(\gamma\) are positive and \(F'\) is nonnegative, the hypothesis of Proposition~\ref{prop.nocyclesforinhibitors} is always satisfied for a single inhibitory population, since in that case \(c < 0\).
\end{remark}

\begin{proof}
Let \(f \colon \dom{1} \to \mathbb{R}^2\) denote the vector field corresponding to the mean-field system~\eqref{eq.meanfield}, so that
\begin{align*}
    \dE{A} & = f_1(\E{A},\E{R}), \\
    \dE{R} & = f_2(\E{A},\E{R}).
\end{align*}
Direct calculations show that
\[
\partial_{\E{A}}f_1(\E{A},\E{R}) = - \beta - \alpha F(\E{B}) + \alpha F'(\E{B}) c \E{S}
\]
and that
\[
\partial_{\E{R}}f_2(\E{A},\E{R}) = - \gamma.
\]
Since \(F\) takes its values between \(0\) and \(1\), it follows that
\beq{eq.divergence}
\nabla \cdot f(\E{A},\E{R}) \leq - \beta - \gamma + \alpha F'(\E{B})c \E{S}.
\eeq
Now, recall that \(\alpha, \beta, \gamma > 0\), that \(\E{S} \geq 0\) since we assume that the state \((\E{A},\E{R}) \in \dom{1}\), and that \(F\) is increasing. Therefore, the divergence~\eqref{eq.divergence} is always negative when \(c \leq 0\). On the other hand, when \(c > 0\), 
\[
\nabla \cdot f(\E{A},\E{R}) \leq - \beta - \gamma + \alpha c \sup F',
\]
which is always negative provided that \(\beta + \gamma > \alpha c \sup F'\). As stated in the remark above, this condition includes the case where \(c \leq 0\). Thus, if \(\beta + \gamma > \alpha c \sup F'\), then the divergence of the vector field is always negative in \(\dom{1}\), and the criterion of \cite{bendixson_sur_1901} guarantees that there are no cycles in \(\dom{1}\).
\end{proof}

In the same way, the Wilson--Cowan system cannot predict chaotic behavior with two populations, since it has only two dimensions while existence of chaotic solutions requires at least three dimensions in continuous dynamical systems. However, there is no dimensional argument to rule out this possibility with our mean-field system, since it has four dimensions for two populations. We will indeed see in section~\ref{sec.ex.chaos} an example of a pair of two excitatory populations that exhibit chaotic behavior.

To further compare the mean-field system and its Wilson--Cowan subsystem, we can add an extra parameter \(\varepsilon\) to the mean-field system~\eqref{eq.meanfield}, and consider the dynamical system
\begin{subequations}
\label{eq.mixed}
\begin{align}
    \dE{A}[J] & = - \beta_J \E{A}[J] + \alpha_J F_J(\E{B}[J]) \E{S}[J], \\
    \varepsilon\dE{R}[J] & = - \gamma_J \E{R}[J] + \beta_J \E{A}[J].
\end{align}
\end{subequations}
The parameter \(\varepsilon\) can then be used to study the transition between the models. First, the system~\eqref{eq.meanfield} corresponds to the case \(\varepsilon = 1\). Then, in the regime where \(0 < \varepsilon \ll 1\), \eqref{eq.mixed} is a slow-fast system with two time scales, where the active fractions of populations are the slow variables whereas the refractory fractions are the fast variables. Ultimately, in the limit where \(\varepsilon\) goes to zero, the fast components can be considered to be at equilibrium, so that each refractory fraction is forced to \(\E{R}[J] = \frac{\beta_J}{\gamma_J} \E{A}[J]\), and we retrieve the Wilson--Cowan subsystem~\eqref{eq.wilsoncowan}. This suggests that the reduction of the mean-field system to the Wilson--Cowan subsystem is valid when the refractory fractions of populations vary much faster than their active fractions, which is the case when the firing rates of populations are small compared to the rates \(\beta_J\) and \(\gamma_J\). However, this approximation is no longer valid for regimes of larger firing rates where the activation can occur on a time scale comparable to the transitions in and out of the refractory state. It follows that our model can be seen as an extension of the Wilson--Cowan model that is still valid when the firing rate cannot be taken as tending to 0.

\section{Examples}
\label{sec.examples}

In this section, we present three examples where the dynamical behavior of the mean-field system~\eqref{eq.meanfield} is different than that of the Wilson--Cowan subsystem~\eqref{eq.wilsoncowan}, where refractory fractions of populations are fixed to their equilibrium solutions. In all cases, we also show a sample trajectory of the Markov chain described in section~\ref{sec.markovchain}. These examples show that there are cases in which the refractory fractions of populations are needed to get an accurate picture of the average behavior of the dynamics on the network. 

For all examples presented here, we will assume that the neurons' activation rate is sigmoidal, with
\beq{eq.defF}
F_J(y) \defeq \frac{1}{1 + \exp\Bigl( - \dfrac{y - \theta_J}{s_{\theta_J}} \Bigr)},
\eeq
where \(\theta_J \in \mathbb{R}\) is a threshold parameter and \(s_{\theta_J} > 0\) is a scaling parameter.

\subsection{A single excitatory population}
\label{sec.ex.cycle1}

Consider a network of \(N\) neurons with a single population, with parameters
\beq{eq.ex.cycle1.params}
\begin{aligned}
    \alpha & = 12.5\>[\gamma], &
    \theta & = 2, \\
    \beta & = 3\>[\gamma], &
    s_\theta & = 0.4, \\
    \gamma & = 1\>[\gamma], &
    Q & = 0, \\
    c & = 8, \\
\end{aligned}
\eeq
where we have dropped subscripts that would refer to the unique population, and where transition rates \(\alpha\), \(\beta\) and \(\gamma\) are measured in units of \(\gamma\). Indeed, each term in the equations of the mean-field system~\eqref{eq.meanfield} is proportional to one of these rates, so setting \(\gamma = 1\) is equivalent to measuring time in units of \(\nicefrac{1}{\gamma}\). We fix an initial state
\beq{eq.ex.cycle1.initstate}
\bigl( \E{A}, \E{R} \bigr)(0) = (0.1, 0.3).
\eeq
Notice that \(\E{R}(0) = \frac{\beta}{\gamma} \E{A}(0)\), so that \(\dE{R}(0) = 0\) and the initial state belongs to the domain \(\dom[w]{1}\).

\begin{figure*}
\includegraphics[width=\linewidth]{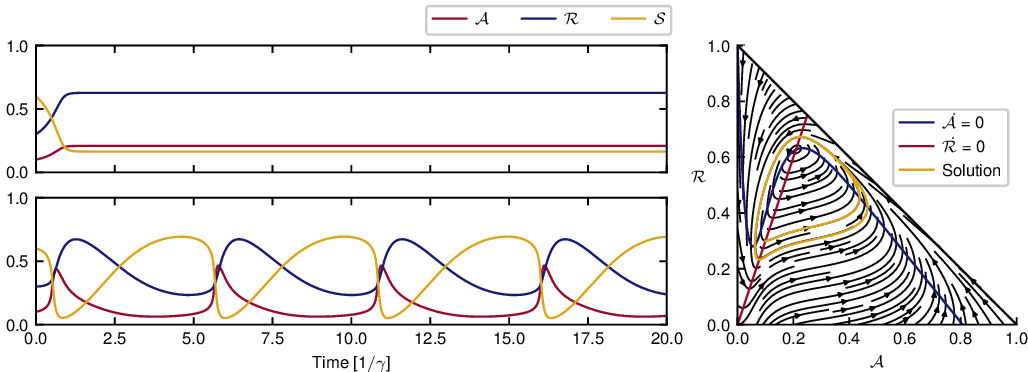}
\caption{Solutions of the dynamical systems~\eqref{eq.meanfield} (bottom left) and~\eqref{eq.wilsoncowan} (top left) with parameters~\eqref{eq.ex.cycle1.params} from the initial state~\eqref{eq.ex.cycle1.initstate}, and phase plane of the system~\eqref{eq.meanfield} (right) with the same parameters. The solution illustrated on the phase plane is the same solution as that on the bottom left panel}
\label{fig.ex.cycle1.solutions}
\end{figure*}

The mean-field system~\eqref{eq.meanfield} and its subsystem~\eqref{eq.wilsoncowan} can be integrated numerically from the initial state~\eqref{eq.ex.cycle1.initstate} with the parameters~\eqref{eq.ex.cycle1.params}. This yields the solutions illustrated in Fig.~\ref{fig.ex.cycle1.solutions}. According to Wilson--Cowan's model, the network's state converges to a stable fixed point. However, according to our mean-field model where the refractory state is explicitly included, the network's state rather converges to a limit cycle. We remark that this cycle is rather robust with respect to the values of the transition rates, which is discussed in detail in appendix~\ref{sec.apx.alphabeta}.

The discrepancy between the long-term behaviors of the mean-field system~\eqref{eq.meanfield} and the Wilson--Cowan system~\eqref{eq.wilsoncowan} shown in Fig.~\ref{fig.ex.cycle1.solutions} suggests that the mixed system~\eqref{eq.mixed} undergoes a supercritical Hopf bifurcation as \(\varepsilon\) varies from \(0\) to \(1\), at the fixed point to which the solution of the Wilson--Cowan system converges. This can be verified numerically by computing the eigenvalues of the Jacobian matrix of the system~\eqref{eq.mixed} with respect to \(\varepsilon\). The results are illustrated in Fig.~\ref{fig.ex.cycle1.eigenvals}. 

\begin{figure}
\includegraphics[width=\linewidth]{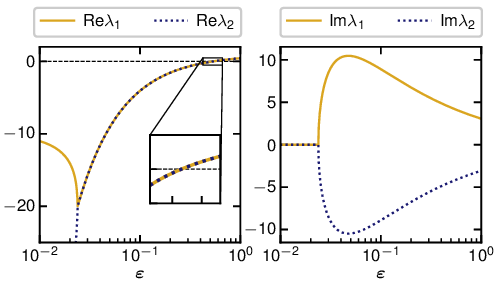}
\caption{Eigenvalues of the Jacobian matrix of the mixed system~\eqref{eq.mixed} with parameters~\eqref{eq.ex.cycle1.params} with respect to \(\varepsilon\), evaluated at the fixed point to which the Wilson--Cowan subsystem converges}
\label{fig.ex.cycle1.eigenvals}
\end{figure}

To obtain a better understanding of the bifurcation, it is instructive to draw a bifurcation diagram. It is possible to obtain a numerical estimate of such a diagram in the \(\varepsilon\)--\(\E{A}\)--\(\E{R}\) space. To do so, we first compute the coordinates of the fixed point simply by computing the zero of the equation~\eqref{eq.wilsoncowan} for the parameters~\eqref{eq.ex.cycle1.params}. Then, for multiple values of \(\varepsilon\), we add a small perturbation to the coordinates of the fixed point, and we integrate numerically the mixed system~\eqref{eq.mixed}. When a stable solution has been reached (either the same fixed point, or a limit cycle around it), we find its coordinates. Plotting these solutions in the \(\varepsilon\)--\(\E{A}\)--\(\E{R}\) space yields the three-dimensional bifurcation diagram illustrated in Fig.~\ref{fig.ex.cycle1.bifurcationdiagram}.

\begin{figure}
\includegraphics[width=\linewidth]{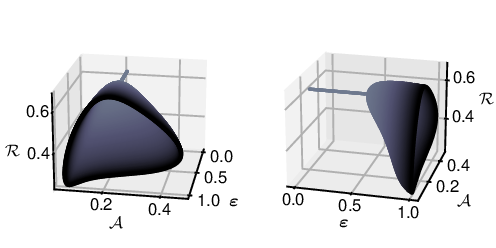}
\caption{Two views of the three-dimensional  bifurcation diagram for the system~\eqref{eq.mixed} with parameters~\eqref{eq.ex.cycle1.params}, where the color of the surface is a function of \(\varepsilon\) to make the surface easier to see}
\label{fig.ex.cycle1.bifurcationdiagram}
\end{figure}

Finally, it is interesting to compare the solutions illustrated in Fig.~\ref{fig.ex.cycle1.solutions} to sample trajectories of the Markov chain that both macroscopic models seek to approximate. To provide a useful comparison, the same parameters~\eqref{eq.ex.cycle1.params} can be used with weights \(W = \nicefrac{c}{N}\) between each pair of neurons, and the initial state can be taken randomly so that a neuron is active at time zero with probability \(0.1\) and refractory with probability \(0.3\). In this way, the microscopic initial state corresponds to~\eqref{eq.ex.cycle1.initstate}. Sample trajectories can be obtained from numerical simulations using the Doob--Gillespie algorithm \citep{gillespie_general_1976}. A typical trajectory obtained with a network of \(N = 2000\) neurons is given in Fig.~\ref{fig.ex.cycle1.trajectory}, where we clearly distinguish oscillations in the network's activity that are analogous to those depicted in Fig.~\ref{fig.ex.cycle1.solutions} (bottom). Therefore, we conclude that in this case, the model~\eqref{eq.meanfield} provides a more accurate prediction of the network's activity than the Wilson--Cowan subsystem~\eqref{eq.wilsoncowan}.

\begin{figure}
\includegraphics[width=\linewidth]{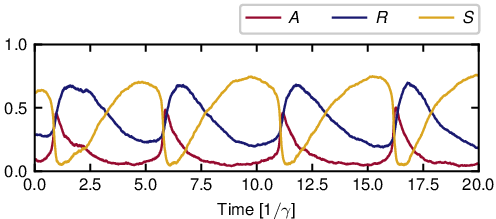}
\caption{Typical trajectory of the Markov chain described in section~\ref{sec.markovchain} with parameters~\eqref{eq.ex.cycle1.params}, from an initial state taken randomly such that the probabilities for a neuron to be active and refractory at time zero are respectively \(0.1\) and \(0.3\).}
\label{fig.ex.cycle1.trajectory}
\end{figure}

\subsection{An excitator--excitator pair}
\label{sec.ex.chaos}

The last example raises an interesting question: if the activities of two excitatory populations can oscillate by themselves, what happens when they are connected together? We give here a possible answer to this question in a case where two such populations are weakly connected to one another.

Consider a network of \(N\) neurons with two excitatory populations, with parameters
\begin{subequations}
\label{eq.ex.chaos.params}
\beq{eq.ex.chaos.params1}
\begin{aligned}
    \alpha_1 & = 12.5\>[\gamma_1], &
    \theta_1 & = 2, \\
    \beta_1 & = 3\>[\gamma_1], &
    s_{\theta_1} & = 0.4, \\
    \gamma_1 & = 1\>[\gamma_1], &
    Q_1 & = 0,
\end{aligned}
\eeq
and
\beq{eq.ex.chaos.params2}
\begin{aligned}
    \alpha_2 & = 3.6\>[\gamma_1], &
    \theta_2 & = 0.84, \\
    \beta_2 & = 8\>[\gamma_1], &
    s_{\theta_2} & = 0.2, \\
    \gamma_2 & = 0.8\>[\gamma_1], &
    Q_2 & = 0,
\end{aligned}
\eeq
where in the same way as in the last example, we measure transition rates in units of \(\gamma_1\) so that time is measured in units of \(\nicefrac{1}{\gamma_1}\). The connections between these populations are described by the matrix
\beq{eq.ex.chaos.c}
    c = \begin{pmatrix}
        8 & 0.6 \\
        0.01 & 14
    \end{pmatrix}.
\eeq
\end{subequations}
We fix an initial state
\beq{eq.ex.chaos.initstate}
(\E{A}[1], \E{A}[2], \E{R}[1], \E{R}[2])(0) = (0.1, 0.02, 0.3, 0.2).
\eeq
Notice that for both populations, \(\E{R}[J](0) = \frac{\beta_J}{\gamma_J} \E{A}[J](0)\), so the initial state belongs to the domain \(\dom[w]{2}\).

Integrating numerically the mean-field system~\eqref{eq.meanfield} and its Wilson--Cowan subsystem from the initial state~\eqref{eq.ex.chaos.initstate} with the parameters~\eqref{eq.ex.chaos.params} yields the solutions illustrated in Fig.~\ref{fig.ex.chaos.solutions}. According to Wilson--Cowan's model, the network's state simply converges to a stable fixed point. However, our mean-field model predicts that the network's activity will exhibit aperiodic behavior, seemingly chaotic. 

\begin{figure}
\includegraphics[width=\linewidth]{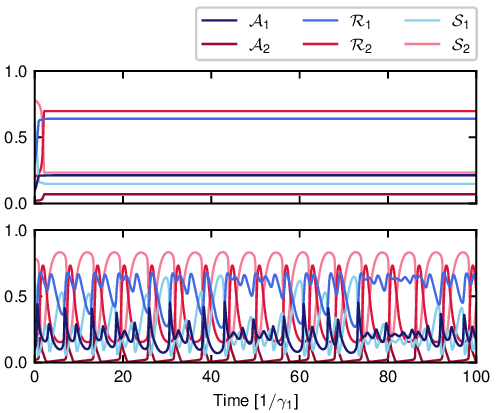}
\caption{Solutions of the dynamical systems~\eqref{eq.meanfield} (bottom) and~\eqref{eq.wilsoncowan} (top) with parameters~\eqref{eq.ex.chaos.params} from the initial state~\eqref{eq.ex.chaos.initstate}}
\label{fig.ex.chaos.solutions}
\end{figure}

To understand the behavior of the network's state, it is instructive to illustrate the solutions of the mean-field system in other ways. In Fig.~\ref{fig.ex.chaos.longsolutions}, the solution of the mean-field system over increasing time intervals is illustrated in the \(\E{A}[1]\)--\(\E{R}[1]\) subspace. The projection of the solution onto this subspace appears not to converge to a point nor to a closed curve, and seems rather to be dense in a bounded subset of the plane. The solution can also be projected onto three-dimensional subspaces. A projection onto the \(\E{A}[1]\)--\(\E{A}[2]\)--\(\E{R}[2]\) subspace is shown in Fig.~\ref{fig.ex.chaos.attractor}. The solution then seems to converge to a bounded subset of lower dimension, which suggests the presence of a strange attractor. 

\begin{figure}
\includegraphics[width=\linewidth]{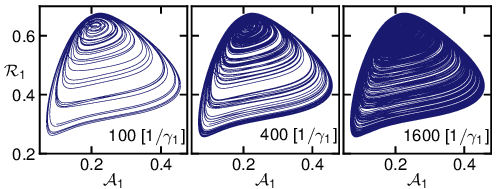}
\caption{Projections of solutions of the dynamical system~\eqref{eq.meanfield} with parameters~\eqref{eq.ex.chaos.params} from the initial state~\eqref{eq.ex.chaos.initstate} over three increasing time intervals}
\label{fig.ex.chaos.longsolutions}
\end{figure}

\begin{figure}
\includegraphics[width=\linewidth]{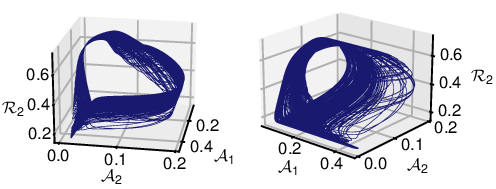}
\caption{Two views of a projection of the solution of the dynamical system~\eqref{eq.meanfield} with parameters~\eqref{eq.ex.chaos.params} from the initial state~\eqref{eq.ex.chaos.initstate} over 1600 time units}
\label{fig.ex.chaos.attractor}
\end{figure}

It is then interesting to estimate a fractal dimension for this attractor. Using the method of \cite{grassberger_measuring_1983}, we estimate its correlation dimension from a solution integrated over 500 000 time units. To do so, we cut the first 1000 time units from the solution to make sure to keep only points on the attractor. Then, we fix a small radius \(r\), and for a point \(x\) in the remaining points, we count the number \(N_x\) of other points lying in a ball of radius \(r\) around \(x\). We do so for 100 such points \(x\), and we average the resulting counts \(N_x\) to find a correlation \(C(r)\). Applying this recipe for many different values of \(r\) between \(10^{-4}\) and \(10^{-2}\), we obtain a power relation of the form
\beq{eq.ex.chaos.cordim}
C(r) \sim r^\nu,
\eeq
where \(\nu\) is the correlation dimension of the attractor. The result is illustrated in Fig.~\ref{fig.ex.chaos.cordim}. To find a value for \(\nu\), we perform a linear regression for \(\log C\) as a function of \(\log r\). The slope of the fitted line is the correlation dimension. We obtain 
\beq{}
\nu = 2.173 \pm 0.005,
\eeq
the error being the standard error from the linear regression. 

\begin{figure}
\includegraphics[width=\linewidth]{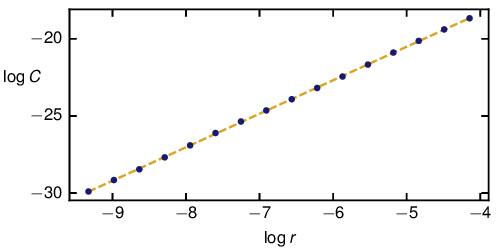}
\caption{Logarithm of the correlation \(C\) as a function of the logarithm of the radius \(r\), to estimate the correlation dimension of the attractor}
\label{fig.ex.chaos.cordim}
\end{figure}

To confirm the chaotic behavior of the solution, we compute the largest Lyapunov exponent of the solution, which is a standard method used to either detect or define chaos \citep{kinsner_characterizing_2006, hunt_defining_2015}. We use the discrete QR method as described by \cite{dieci_computation_1997}. We do so for four different combinations of integration interval and time steps. The results are given in Table~\ref{tab.ex.chaos.lyapunov}. The largest Lyapunov exponent of the system is positive, indicating a chaotic behavior.

\begin{table}[t]
\centering
\begin{tabular}{*{3}{>{\centering\arraybackslash}p{14mm}}}
    Integration interval & Time step & Lyapunov exponent \\
    \([\nicefrac{1}{\gamma_1}]\) & \([\nicefrac{1}{\gamma_1}]\) & \\
    \hline
    1000  & 0.01  & 0.1592 \\
    10000 & 0.01  & 0.1572 \\
    1000  & 0.001 & 0.1691 \\
    10000 & 0.001 & 0.1633
\end{tabular}
\caption{Values of the largest Lyapunov exponent for numerical solutions of the dynamical system~\eqref{eq.meanfield} with parameters~\eqref{eq.ex.chaos.params} and initial state~\eqref{eq.ex.chaos.initstate}, for four different combinations of integration interval and time step}
\label{tab.ex.chaos.lyapunov}
\end{table}

Finally, we compare the solutions illustrated in Fig.~\ref{fig.ex.chaos.solutions} to sample trajectories of the Markov chain to determine how well the macroscopic models approximate the behavior of the network. To do so, we use the parameters~\eqref{eq.ex.chaos.params} with weights \(W_{JK} = \nicefrac{c_{JK}}{\abs{K}}\) from the neurons of population \(K\) to neurons of population \(J\). We choose randomly an initial state to each neuron of population \(J\) so that it is active with probability \(\E{A}[J](0)\) and refractory with probability \(\E{R}[J](0)\), where the macroscopic initial values are taken from the initial state~\eqref{eq.ex.chaos.initstate}. As is in the first example, sample trajectories are obtained using the Doob--Gillespie algorithm. A typical result for a network of \(N = 2000\) neurons, with \(1000\) neurons in each population, is given in Fig.~\ref{fig.ex.chaos.trajectory}. As seen on this figure, the evolution of the network's state does exhibit the aperiodic behavior predicted by our mean-field model. Therefore, we conclude that in this case, including the refractory state explicitly in the dynamical system leads to a more accurate prediction of the network's activity than to force it to its equilibrium solution.

\begin{figure}
\includegraphics[width=\linewidth]{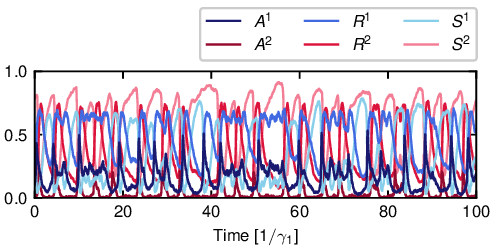}
\caption{Typical trajectory of the Markov chain described in section~\ref{sec.markovchain} with parameters~\eqref{eq.ex.chaos.params}, from an initial state taken randomly so that the probabilities for a neuron to be active or refractory correspond to the macroscopic initial state~\eqref{eq.ex.chaos.initstate}}
\label{fig.ex.chaos.trajectory}
\end{figure}

\subsection{An excitator--inhibitor pair}

In this last example, we show that the benefit of including the refractory state explicitly in the dynamical system is not only to allow for new dynamical behaviors, but also to predict more accurately the dynamical behavior of the underlying Markov chain in other situations.

Consider a network of \(N\) neurons with two populations: one excitatory labelled \(E\), with parameters
\begin{subequations}
\label{eq.ex.cycle2.params}
\beq{eq.ex.cycle2.paramsE}
\begin{aligned}
    \alpha_E & = 10\>[\gamma_I], &
    \theta_E & = 0, \\
    \beta_E & = 0.8\>[\gamma_I], &
    s_{\theta_E} & = 0.4, \\
    \gamma_E & = 4\>[\gamma_I], &
    Q_E & = 0,
\end{aligned}
\eeq
and one inhibitory labelled \(I\), with parameters
\beq{eq.ex.cycle2.paramsI}
\begin{aligned}
    \alpha_I & = 9\>[\gamma_I], &
    \theta_I & = 3, \\
    \beta_I & = 1\>[\gamma_I], &
    s_{\theta_I} & = 0.4, \\
    \gamma_I & = 1\>[\gamma_I], &
    Q_I & = 0,
\end{aligned}
\eeq
where we measure transition rates according to \(\gamma_I\) so that time is measured in units of \(\nicefrac{1}{\gamma_I}\). The connections between these populations are described by the matrix
\beq{eq.ex.cycle2.c}
    c = \begin{pmatrix}
        c_{EE} & c_{EI} \\
        c_{IE} & c_{II}
    \end{pmatrix} = \begin{pmatrix}
        8 & -12 \\
        9 & -2
    \end{pmatrix}.
\eeq
\end{subequations}
We fix an initial state
\beq{eq.ex.cycle2.initstate}
\bigl( \E{A}[E], \E{A}[I], \E{R}[E], \E{R}[I] \bigr)(0) = (0.4, 0.4, 0.08, 0.4),
\eeq
which belongs to the domain \(\dom[w]{2}\).

Integrating numerically the mean-field system~\eqref{eq.meanfield} and its subsystem~\eqref{eq.wilsoncowan} with the parameters~\eqref{eq.ex.cycle2.params} from the initial state~\eqref{eq.ex.cycle2.initstate} yields the solutions presented in Fig.~\ref{fig.ex.cycle2.solutions}. According to Wilson--Cowan's model, the network's state converges to a stable fixed point, but according to our mean-field model, it rather converges to a limit cycle. However, we remark that Wilson--Cowan's model can also predict oscillations with parameters close to those chosen here: for example, if the connection matrix is changed to \(c = \bigl(\begin{smallmatrix} 9 & -12 \\ 9 & -1 \end{smallmatrix}\bigr)\), then both models predict oscillations. Hence, expanding Wilson--Cowan's model to our mean-field model modifies the values of parameters for which oscillations are predicted. 

\begin{figure}
\includegraphics[width=\linewidth]{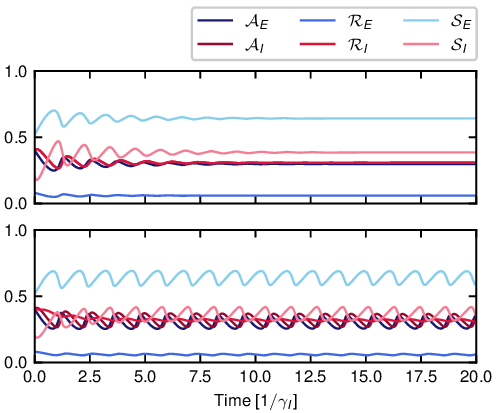}
\caption{Solutions of the dynamical systems~\eqref{eq.meanfield} (bottom) and~\eqref{eq.wilsoncowan} (top) with parameters~\eqref{eq.ex.cycle2.params} from the initial state~\eqref{eq.ex.cycle2.initstate}}
\label{fig.ex.cycle2.solutions}
\end{figure}

As in the first example, the difference between the long-term behaviors of the mean-field system and its Wilson--Cowan subsystem suggests a bifurcation with respect to \(\varepsilon\) in the mixed system~\eqref{eq.mixed}. This is verified numerically by computing the eigenvalues of the Jacobian matrix of the mixed system with respect to \(\varepsilon\). Doing so indeed shows that as \(\varepsilon\) goes from 0 to 1, the real part of a pair of conjugate eigenvalues goes from negative to positive. Hence, we see that the mixed system undergoes a Hopf bifurcation in this interval.

It is possible to further understand this bifurcation by drawing a bifurcation diagram. It is unfortunately not possible to draw a complete bifurcation diagram due to the dimension of the system, but it is still possible to obtain a diagram for each individual state component. The result is illustrated in Fig.~\ref{fig.ex.cycle2.bifurcationdiagram}, where the maximum and minimum values of each state component on the cycle or on the fixed point is plotted against \(\varepsilon\). According to the results, the Hopf bifurcation appears to be supercritical.

\begin{figure}
\includegraphics[width=\linewidth]{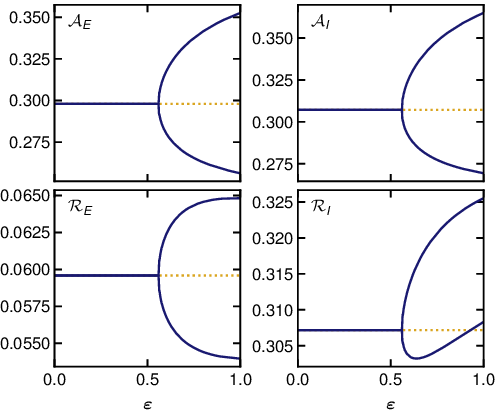}
\caption{Bifurcation diagrams of the mixed system~\eqref{eq.mixed} with parameters~\eqref{eq.ex.cycle2.params}, where the dotted yellow lines represent the coordinates of the fixed point and the blue lines represent either the maximum and minimum values of the component on the cycle, or the value of the component at the fixed point}
\label{fig.ex.cycle2.bifurcationdiagram}
\end{figure}

Finally, the solutions of the dynamical systems can be compared to trajectories of the Markov chain whose macroscopic behavior is approximated by both models. To do so, the same parameters~\eqref{eq.ex.cycle2.params} can be used with weights \(W_{JK} = \nicefrac{c_{JK}}{\abs{K}}\) from neurons of population \(K\) to neurons of population \(J\). Then, the initial state of each neuron of population \(J\) is taken randomly so that it is active with probability \(\E{A}[J](0)\) and refractory with probability \(\E{R}[J](0)\), where macroscopic initial values are those given in the initial state~\eqref{eq.ex.cycle2.initstate}. As in the other examples, sample trajectories are obtained using the Doob--Gillespie algorithm. A typical trajectory with a network of \(N = 2000\) neurons with \(1000\) neurons in each population is given in Fig.~\ref{fig.ex.cycle2.trajectory}. This trajectory exhibits distinct oscillations in the network's activity. Thus, it is clear that in this case, the mean-field model provides a more accurate approximation of the network's behavior than its Wilson--Cowan subsystem.

\begin{figure}
\includegraphics[width=\linewidth]{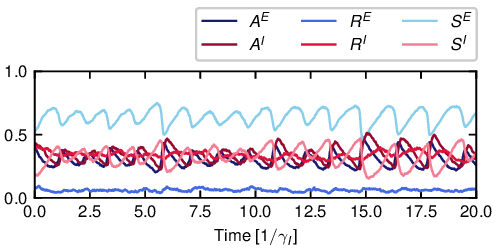}
\caption{Typical trajectory of the Markov chain described in section~\ref{sec.markovchain} with parameters~\eqref{eq.ex.cycle2.params}, from an initial state taken randomly such that the probabilities for a neuron to be active or refractory at time zero correspond to the macroscopic initial state~\eqref{eq.ex.cycle2.initstate}}
\label{fig.ex.cycle2.trajectory}
\end{figure}

\section{Conclusion}

The Wilson--Cowan model has played an important role in the description of neural systems at the macroscopic level. It has been shown that when considering several excitatory and inhibitory populations, the Wilson--Cowan model can exhibit rich dynamics such as oscillations, bistability and chaos. Furthermore, it is a useful tool in better understanding biological neural networks, especially sensory systems.

One of the assumptions on which the Wilson--Cowan model relies is that the ratio of the numbers of active and refractory neurons is constant in a single population. In this work, we showed that lifting this assumption can reveal novel dynamics in the model, such as oscillations in the activity of a single population or chaotic behavior in the activity of two populations.

An interesting byproduct of our method is that when constructing the model as done in section~\ref{sec.macromodel}, the way in which the dynamics might be affected by correlations between the activities of different populations becomes quite clear. Indeed, to obtain a closed dynamical system, we chose to neglect covariances in equations~\eqref{eq.macrosystem}, which led us to approximate the expectations \(\expect{F_J(B_J^t)S_t^J}\) by the corresponding functions where the variables \(S_t^J\) and \(B_t^J\) are replaced with their expectations. Using a higher-order moment closure, we could take into account the correlations between activities of different populations. We intend to investigate this in future work.

\backmatter
\bmhead{Acknowledgments}

This work was supported by the Natural Sciences and Engineering Research Council of Canada (N.D, P.D, V.P), the Fonds de recherche du Qu\'ebec -- Nature and technologies (N.D., V.P.), and the Sentinel North program of Universit\'e Laval (N.D., P.D), funded by the Canada First Research Excellence Fund.

\bmhead{Code availability}

All numerical results presented in this paper were obtained with the PopNet package \citep{painchaud_popnet_2022}, written in the Python programming language and available on GitHub.

\bibliography{references}

\appendix

\section{Alternative construction of the model}
\label{sec.apx.construction}

In this section, we provide an alternative construction of the model presented in section~\ref{sec.model}, in which we allow parameters to be random, independent and identically distributed over populations. As long as we keep the mean-field assumption, the resulting macroscopic is still the dynamical system~\eqref{eq.meanfield}. This construction is presented in greater detail and for a more general Markov chain in chapter~2 of \citep{painchaud_dynamique_2021}; the special case of the Markov chain presented here is discussed in section~4.2 of the latter reference. For basic definitions and results about stochastic processes and Markov chains, we refer to the books by \cite{doob_stochastic_1990} and \cite{norris_markov_1997}. 

We start from a similar setting as that described at the beginning of section~\ref{sec.model}. We consider a network of \(N\) neurons labelled from integers from 1 to \(N\), and split into a collection \(\mathscr{P}\) of populations, which are assumed to be made of a large number of neurons. The weight matrix \(W\) is now random, defined on a probability space \((H, \mathscr{H}, \mu)\), where \(\mathscr{H}\) and \(\mu\) respectively are a \(\sigma\)-algebra and a probability measure on the set \(H\). We still interpret links between neurons as being either excitatory or inhibitory according to the sign of the corresponding entry of \(W\). Then, for each neuron \(j\), we consider real-valued random variables \(\alpha_j\), \(\beta_j\), \(\gamma_j\) and \(\theta_j\), defined on \((H,\mathscr{H},\mu)\) as well, the first three being nonnegative. We assume that all parameters \(\alpha_j\), \(\beta_j\), \(\gamma_j\), \(\theta_j\) and \(W_{jk}\) are independent random variables, and that they are identically distributed over populations in \(\mathscr{P}\).

We still see the parameters \(\beta_j\) and \(\gamma_j\) as being the transition rates of neuron \(j\) as in Fig.~\ref{fig.rates}. For the reduction to the macroscopic model to be possible, we will use a different activation rate:
\[
a_j(\eta, x) \defeq \alpha_j \charf{T_j(x)}(\eta)
\]
where \(\charf{T_j(x)}\) is the indicator function of the set
\[
T_j(x) \defeq \Bigl\{ \eta \in H : \sum_{k=1}^N W_{jk}(\eta) \Re x_k + Q_J > \theta_j(\eta) \Bigr\},
\]
where the external input \(Q_J\) is still assumed to be constant over population \(J\). 

Now, any \(\eta \in H\) fixes a choice of parameters, allowing to construct a Markov chain as in section~\ref{sec.markovchain}. To do so, we define the matrix \(M^\eta \defeq \{m^\eta(x,y) : x,y \in E\}\) with entries
\[
m^\eta(x,y) \defeq \sum_{j=1}^N m_j^\eta(x,y) \prod_{\substack{k=1 \\k\neq j}}^N \delta_{x_ky_k}
\]
where
\[
\begin{aligned}
m_j^\eta(x,y) \defeq a_j(\eta, x) & (1 - \abs{x_j}) \bigl( \Re y_j - (1 - \abs{y_j}) \bigr) \\
    {} + \beta_j(\eta) & \Re(x_j) (\Im y_j - \Re y_j) \\
    {} + \gamma_j(\eta) & \Im(x_j) \bigl( (1 - \abs{y_j}) - \Im y_j \bigr).
\end{aligned}
\]
Again, \(M^\eta\) is the generator of a continuous-time Markov chain. Hence, it follows from the Kolmogorov extension theorem that a probability measure \(\mathbb{P}^\eta\) exists on \((\Omega, \mathscr{F}) \defeq (E,2^E)^{[0,\infty)}\) such that for any \(x,y \in E\), as \(\Delta t \downarrow 0\),
\[
\parambprob{X_{t+\Delta t}=y \given X_t=x} = \delta_{xy} + m^\eta(x,y) \Delta t + o(\Delta t),
\]
where \(\{X_t\}_{t\geq 0}\) is the coordinate mapping process on \((\Omega,\mathscr{F})\), defined as \(X_t(\omega) \defeq \omega(t)\). Similar relations to \eqref{eq.transitionprob} with a dependence on \(\eta\) also hold. The measures \(\mathbb{P}^\eta\) can be combined in a single measure \(\mathbb{P}\) on the product space \((H \times \Omega, \mathscr{H} \otimes \mathscr{F})\) such that
\beq{eq.probmeasure}
\expect{Z} = \int_H \int_\Omega Z(\eta,\omega) \d\mathbb{P}^\eta(\omega) \d\mu(\eta)
\eeq
for any random variable \(Z\) on \((H \times \Omega, \mathscr{H} \otimes \mathscr{F})\), where \(\mathbb{E}\) is the expectation with respect to \(\mathbb{P}\). 

To study the average behavior of the Markov chain, we start by defining
\begin{align*}
    p_j^\eta(t) & \defeq \bprob{X_t^j = 1 \given \eta}, \\
    q_j^\eta(t) & \defeq \bprob{X_t^j = 0 \given \eta}, \\
    r_j^\eta(t) & \defeq \bprob{X_t^j = i \given \eta},
\end{align*}
where we condition on the projection onto \(H\). Then, in the same way as in section~\ref{sec.dynsys}, it can be shown that
\begin{align*}
    \dot{p}_j^\eta(t) & = - \beta_j(\eta) p_j^\eta(t) + \expect[\big]{\bigl( 1 - \abs{X_t^j} \bigr) a_j(\eta, X_t) \given \eta}, \\
    \dot{q}_j^\eta(t) & = - \expect[\big]{\bigl( 1 - \abs{X_t^j} \bigr) a_j(\eta, X_t) \given \eta} + \gamma_j(\eta) r_j^\eta(t), \\
    \dot{r}_j^\eta(t) & = - \gamma_j(\eta) r_j^\eta(t) + \beta_j(\eta) p_j^\eta(t).
\end{align*}
Now, we can define
\begin{align*}
    p_j(t) & \defeq \bprob{X_t^j = 1}, \\
    q_j(t) & \defeq \bprob{X_t^j = 0}, \\
    r_j(t) & \defeq \bprob{X_t^j = i}.
\end{align*}
Then \(\dot{p}_j(t) = \frac{\d}{\d t} \int_H p_j^\eta(t) \d\mu(\eta)\). The functions \(\eta \mapsto p_j^\eta(t)\) can be dominated by the function \(\sum_{x\in E} \abs{\beta_j + \alpha_j \charf{T_j(x)}}\), which is integrable over \(H\), so that the derivative can be passed under the integral sign. Similar bounds hold for \(q_j^\eta(t)\) and \(r_j^\eta(t)\), and it follows that
\begin{align*}
    \dot{p}_j(t) & = - \expect{\beta_j \Re X_t^j} + \expect[\big]{\bigl( 1 - \abs{X_t^j} \bigr) a_j(\cdot, X_t)}, \\
    \dot{q}_j(t) & = - \expect[\big]{\bigl( 1 - \abs{X_t^j} \bigr) a_j(\cdot, X_t)} + \expect{\gamma_j \Im X_t^j}, \\
    \dot{r}_j(t) & = - \expect{\gamma_j \Im X_t^j} + \expect{\beta_j \Re X_t^j}.
\end{align*}

To pass to the macroscopic point of view, we consider for \(J \in \mathscr{P}\) the fractions of population \(A_t^J\), \(R_t^J\) and \(S_t^J\) defined as in section~\ref{sec.macromodel}, that is,
\begin{align*}
A_t^J & \defeq \frac{1}{\abs{J}} \sum_{j\in J} \Re X_t^j, \\
R_t^J & \defeq \frac{1}{\abs{J}} \sum_{j\in J} \Im X_t^j, \\
S_t^J & \defeq \frac{1}{\abs{J}} \sum_{j\in J} \bigl( 1 - \abs{X_t^j} \bigr).
\end{align*}
Now that parameters are random, we cannot simply sum over populations to obtain the differential equations \eqref{eq.macrosystem} from the expressions we obtained for the derivatives of \(p_j\), \(q_j\) and \(r_j\). However, we can still argue that the same macroscopic dynamical system describes the average behavior of the network's activity, at least approximatively. To do so, we define the subpopulations
\[
J_t^\xi \defeq \{j \in J : X_t^j = \xi\}
\]
for \(\xi \in \{0,1,i\}\). The size of these subpopulations can easily be related to the macroscopic state variables. Indeed, since \(\Re X_t^j = 1\) if and only if \(X_t^j = 1\), then \(\abs{J_t^1} = \abs{J} A_t^J\). Similarly, \(\abs{J_t^0} = \abs{J} S_t^J\) and \(\abs{J_t^i} = \abs{J} R_t^J\). Now, notice that
\[
\frac{1}{\abs{J}} \sum_{j\in J} \beta_j \Re X_t^j
    = \frac{1}{\abs{J}} \sum_{j\in J_t^1} \beta_j
    = \frac{A_t^J}{\abs{J_t^1}} \sum_{j\in J_t^1} \beta_j.
\]
Assuming that the random variables \(\beta_j\) are independent and identically distributed over populations, if the size of \(J\) is very large, the law of large numbers leads to expect the average \(\frac{1}{\abs{J}} \sum_{j\in J} \beta_j\) to be very close to the expectation \(\beta_J \defeq \expect{\beta_j}\). In the same way, since the value of \(\beta_j\) should not influence the probability that \(j\) is active at time \(t\), we expect the average \(\frac{1}{\abs{J_t^1}} \sum_{j\in J_t^1} \beta_j\) to be very close to \(\beta_J\). This leads to approximate
\beq{eq.approxbeta}
\frac{1}{\abs{J}} \sum_{j\in J} \beta_j \Re X_t^j
    \approx \beta_J A_t^J,
\eeq
for a large population \(J\). In the same way, we approximate
\beq{eq.approxgamma}
\frac{1}{\abs{J}} \sum_{j \in J} \gamma_j \Im X_t^j
    = \frac{R_t^J}{\abs{J_t^i}} \sum_{j \in J_t^i} \gamma_j
    \approx \gamma_J R_t^J
\eeq
where \(\gamma_J \defeq \expect{\gamma_j}\) for \(j \in J\), and
\begin{multline*}
\sum_{k=1}^N W_{jk} \Re X_t^k
    = \sum_{K\in\mathscr{P}} \sum_{k \in K_t^1} W_{jk} \\
    \approx \sum_{K\in\mathscr{P}} \abs{K} A_t^K W_{JK}
    = \sum_{K\in\mathscr{P}} c_{JK} A_t^K
\end{multline*}
where \(c_{JK} \defeq \abs{K} W_{JK}\) and \(W_{JK} \defeq \expect{W_{jk}}\) for \(j \in J\) and \(k \in K\). The last case requires more attention. Using the last approximation, we start by approximating
\begin{align}
\notag
\frac{1}{\abs{J}} \sum_{j\in J} & \bigl( 1 - \abs{X_t^j} \bigr) a_j(\cdot, X_t) \\
    \notag
     & = \frac{S_t^J}{\abs{J_t^0}} \sum_{j\in J_t^0} \alpha_j \charf{\bigl\{ \sum_{k=1}^N W_{jk} \Re X_t^k + Q_J > \theta_j\bigr\}} \\
     \notag
     & \approx \frac{S_t^J}{\abs{J_t^0}} \sum_{j\in J_t^0} \alpha_j \charf{\{B_t^J > \theta_j\}},
\intertext{where \(B_t^J \defeq \sum_{K\in\mathscr{P}} c_{JK} A_t^K + Q_J\). Using the same argument as in the last cases, we approximate by replacing the rate \(\alpha_j \charf{\{B_t^J > \theta_j\}}\) with its expectation given the state of the network. Since \(\alpha_j\) and \(\theta_j\) are independent, we obtain}
\frac{1}{\abs{J}} \sum_{j\in J} & \bigl( 1 - \abs{X_t^j} \bigr) a_j(\cdot, X_t) \approx \alpha_J F_J(B_t^J) S_t^J
\label{eq.approxalpha}
\end{align}
where \(\alpha_J \defeq \expect{\alpha_j}\) for \(j \in J\) and where \(F_J\) denotes the cumulative distribution function of \(\theta_j\) for \(j \in J\), which is the expectation of \(\charf{\{B_t^J > \theta_j\}}\) given the state of the network. We remark that this approximation is based on idea that the set \(J_t^0\) is large, and that the approximate activation rate \(\alpha_j \charf{\{B_t^J > \theta_j\}}\) has no influence on the probability for neuron \(j\) to be sensitive at time \(t\). However, one could expect this to be false, especially if the standard deviation of \(\theta_j\) is high: if the threshold \(\theta_j\) is higher, it is harder for \(j\) to reach its threshold input and to activate. Therefore, we expect this last approximation to be valid only when the standard deviation of \(\theta_j\) is small. 

Now, consider the macroscopic state variables \(\E{A}[J](t) \defeq \expect{A_t^J}\), \(\E{R}[J](t) \defeq \expect{R_t^J}\) and \(\E{S}[J](t) \defeq \expect{S_t^J}\). By linearity of the derivative,
\[
\dE{A}[J] = \frac{1}{\abs{J}} \sum_{j\in J} \dot{p}_j.
\]
Hence, the approximations \eqref{eq.approxbeta} and \eqref{eq.approxalpha} lead to approximate
\begin{align*}
\dE{A}[J](t)
    & \approx \expect[\big]{-\beta_J A_t^J + \alpha_J F_J(B_t^J) S_t^J} \\
    & = - \beta_J \E{A}[J](t) + \alpha_J \expect{F_J(B_t^J) S_t^J}.
\end{align*}
In the same way, the approximations \eqref{eq.approxbeta}, \eqref{eq.approxgamma} and \eqref{eq.approxalpha} lead to
\begin{align*}
\dE{R}[J](t)
    & \approx - \gamma_J \E{R}[J](t) + \beta_J \E{A}[J](t)
\shortintertext{and}
\dE{S}[J](t)
    & \approx - \alpha_J \expect{F_J(B_t^J) S_t^J} + \gamma_J \E{R}[J](t).
\end{align*}
Thus, the dynamical system \eqref{eq.macrosystem} describes approximatively the evolution of the network's activity. In the same way as in section~\ref{sec.macromodel}, the mean-field assumption yields the dynamical system \eqref{eq.meanfield}.

We conclude by noticing that, even though the approximations we made here seem reasonable, it would be far from obvious to make the argument fully rigorous by properly taking a limit as the sizes of the populations grow infinitely large. This is because the probability measure \(\mathbb{P}\) defined in \eqref{eq.probmeasure} depends on the measures \(\mathbb{P}^\eta\), which are defined from the generators \(M^\eta\) and thus depend on the network's size. Hence, to vary the network's size, we would need to vary the probability measure, and it is not clear how these different measures can be related. Moreover, the sets over which averages of parameters are approximated to their expectations are subpopulations of neurons that are in a given state at a given time, which are random. Thus, to use correctly some sort of law of large numbers, we would need to ensure that these subpopulations become arbitrarily large with a high probability.

\section{Further bifurcation analysis of Example~\ref{sec.ex.cycle1}}
\label{sec.apx.alphabeta}

In this section, we investigate the impact of the transition rates on the dynamics of the single-population network of Example~\ref{sec.ex.cycle1}, with the idea to understand how robust is the limit cycle with respect to changes in their values. As before, we measure all three transition rates \(\alpha\), \(\beta\) and \(\gamma\) in units of \(\gamma\), which means that we keep \(\gamma = 1\). Thus, we are really only interested in modifying \(\alpha\) and \(\beta\). We keep the other parameters fixed, with values given by~\eqref{eq.ex.cycle1.params}.

\begin{figure*}
\includegraphics[width=\linewidth]{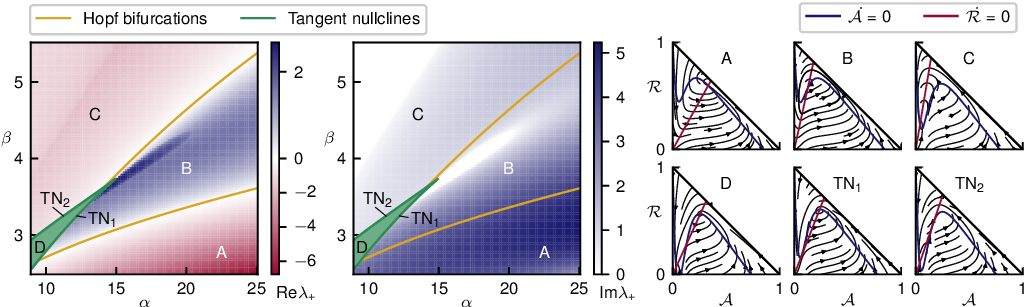}
\caption{Bifurcation diagram on the parameters \(\alpha\) and \(\beta\) for the system~\eqref{eq.meanfield} with parameters~\eqref{eq.ex.cycle1.params}. In regions \textsf{A}, \textsf{B} and \textsf{C} of the left panels, the colors represent the real and imaginary parts of the eigenvalue \(\lambda_+\) of the Jacobian matrix of the system evaluated at the unique fixed point, where \(\lambda_+\) is the eigenvalue with maximal real or imaginary part. The abbreviation \textsf{TN} stands for ``tangent nullclines''. The right panels give sketches of the phase plane in different regions of the bifurcation diagram}
\label{fig.apx.alphabeta}
\end{figure*}

We start by studying the number of fixed points of the system. We do so from the nullclines
\begin{align*}
\E{R}\Bigr\rvert_{\dE{A}=0} & = 1 - \E{A} - \frac{\beta\E{A}}{\alpha F(\E{B})}
\shortintertext{and}
\E{R}\Bigr\rvert_{\dE{R}=0} & = \frac{\beta}{\gamma}\, \E{A}.
\end{align*}
An illustration of these nullclines in the case where \(\alpha\) and \(\beta\) have the values of Example~\ref{sec.ex.cycle1} is given on the phase plane on the right panel of Fig.~\ref{fig.ex.cycle1.solutions}. The fixed points of dynamical system are the points where these nullclines intersect, so they correspond to the zeros of
\[
g(\E{A}) \defeq 1 - \E{A} - \frac{\beta\E{A}}{\alpha F(\E{B})} - \frac{\beta}{\gamma}\E{A}.
\]
Notice that \(g(0) = 1\) while
\[
g\Bigl( \frac{\gamma}{\beta + \gamma} \Bigr) = - \frac{\beta\gamma}{\alpha(\beta + \gamma) F(\frac{c\gamma}{\beta+\gamma} + Q)} < 0.
\]
Thus, the intermediate value theorem shows that the system has at least one fixed point \((\E{A}^*,\E{R}^*)\) with \(\E{A}^* \in [0,\frac{\gamma}{\beta+\gamma}]\). At such a fixed point, the refractory fraction is \(\E{R}^* = \frac{\beta}{\gamma} \E{A}^* = \frac{\beta}{\beta+\gamma}\), and \((\E{A}^*,\E{R}^*) \in \dom{1}\). 

Now, assuming that \(F\) is given by the sigmoidal function~\eqref{eq.defF}, it has the property that
\[
F'(\E{B}) = \frac{1}{s_\theta} F(\E{B}) \bigl(1 - F(\E{B}) \bigr).
\]
Using this result, it is straightforward to compute
\[
g''(\E{A}) = \frac{\beta c}{\alpha s_\theta} \biggl( \frac{1}{F(\E{B})} - 1 \biggr) \biggl( 2 - \frac{c\E{A}}{s_\theta} \biggr).
\]
Since \(F < 1\), it follows that \(g''\) has a single zero. Hence, \(g\) has a single inflection point, and it cannot have more than three zeros.

The above discussion shows that the system has between one and three fixed points. It would be hard to determine precisely the location and the number of fixed points analytically, but it can be done numerically. Indeed, for a given value of \(\alpha\), it is possible to compute the values \(\beta\) for which the nullclines are tangent to each other. These values of \(\beta\) corresponds to the boundary between the domains of the \(\alpha\)--\(\beta\) space where the system admits one and three fixed points. Then, when there is only one fixed point, it is possible to find it numerically, and to compute the eigenvalues of the Jacobian matrix of the system to determine its stability. The result is illustrated in Fig.~\ref{fig.apx.alphabeta}. The regions \textsf{A}, \textsf{B} and \textsf{C} of the bifurcation diagram are those where the system admits a single fixed point, and region \textsf{D} is that where it admits three fixed points. The boundary of region \textsf{D} corresponds to the points where the nullclines are tangent. 

Two important conclusions can be understood from these numerical results. First, in region \textsf{B}, the system admits a single fixed point which is unstable, since at least one eigenvalue of the Jacobian matrix of the system has positive real part. Since the domain \(\dom{1}\) is always invariant, it then follows from the theorem of Poincaré--\cite{bendixson_sur_1901} that a limit cycle exists somewhere in \(\dom{1}\). Moreover, on the curves that bound region \textsf{B}, where the real part of the eigenvalue \(\lambda_+\) is zero, the imaginary part of \(\lambda_+\) is nonzero. Hence, on these curves, the eigenvalues of the Jacobian matrix are conjugates, and their real part changes sign. This implies that the system undergoes Hopf bifurcations on these curves. 

Therefore, if the values of \(\alpha\) and \(\beta\) lie in region \textsf{B} in Fig~\ref{fig.apx.alphabeta}, then a limit cycle exists. We remark however that, a priori, this condition is sufficient but not necessary. Indeed, if the Hopf bifurcations on the bounds of region \textsf{B} are subcritical, the limit cycle may very well continue to exist outside of region \textsf{B}. Further analysis would be needed to determine with more precision the values of \(\alpha\) and \(\beta\) which yield a limit cycle.

\end{document}